\documentclass[pra, twocolumn, preprintnumbers, amsmath,amssymb]{revtex4}

\usepackage{graphicx, color, graphpap}
\usepackage{amsmath}
\usepackage{enumitem}
\usepackage{amssymb}
\usepackage{amsthm}
\usepackage{pstricks}
\usepackage{float}
\usepackage{multirow}
\usepackage{hyperref}
\usepackage[T1]{fontenc}

\long\def\ca#1\cb{} 

\newcommand{\ket}[1]{|#1\rangle}               
\newcommand{\bra}[1]{\langle #1|}              
\newcommand{\dya}[1]{\ket{#1}\!\bra{#1}}
\newcommand{\dyad}[2]{\ket{#1}\!\bra{#2}}        
\newcommand{\ip}[2]{\langle #1|#2\rangle}      

\newcommand{\Hbb}{\mathbb{H}}

\newcommand{\CC}{\mathcal{C}}

\newcommand{\EC}{\mathcal{E}}

\newcommand{\HC}{\mathcal{H}}
\newcommand{\IC}{\mathcal{I}}

\newcommand{\ZC}{\mathcal{Z}}
\newcommand{\ddd}{d}
\newcommand{\bin}{\text{bin}}

\newcommand{\sgt}{\tilde{\sigma}}
\newcommand{\rhoh}{\widehat{\rho}}
\newcommand{\Tr}{{\rm Tr}}

\renewcommand{\geq}{\geqslant}
\renewcommand{\leq}{\leqslant}
\newcommand{\mte}[2]{\langle#1|#2|#1\rangle }
\newcommand{\mted}[3]{\langle#1|#2|#3\rangle }
\newcommand{\eqprop}[2]{\stackrel{\tiny{#1}}{#2}}

\newcommand{\ot}{\otimes}
\newcommand{\ad}{^\dagger}

\newcommand*{\id}{\openone}

\newcommand{\al}{\alpha }
\newcommand{\bt}{\beta }

\newcommand{\dl}{\delta }

\newcommand{\sg}{\sigma }

\newcommand{\om}{\omega }

\newtheoremstyle{example}{\topsep}{\topsep}%
{}
{}
{\bfseries}
{.}
{   }
{\thmname{#1}\thmnumber{ #2}}
\theoremstyle{example}

\newtheorem{theorem}{Theorem}
\newtheorem{lemma}[theorem]{Lemma}

\newtheorem{proposition}[theorem]{Proposition}

\theoremstyle{definition}

\begin{document}

\title{Role of complementarity in superdense coding}

\author{Patrick J. Coles}
\affiliation{Centre for Quantum Technologies, National University of Singapore, 2 Science Drive 3, 117543 Singapore.}

\begin{abstract}
The complementarity of two observables is often captured in uncertainty relations, which quantify an inevitable tradeoff in knowledge. Here we study complementarity in the context of an information processing task: we link the complementarity of two observables to their usefulness for superdense coding (SDC). In SDC, Alice sends two classical dits of information to Bob by sending a single qudit. However, we show that encoding with commuting unitaries prevents Alice from sending more than one dit per qudit, implying that complementarity is necessary for SDC to be advantagous over a classical strategy for information transmission. When Alice encodes with products of Pauli operators for the $X$ and $Z$ bases, we quantify the complementarity of these encodings in terms of the overlap of the $X$ and $Z$ basis elements. Our main result explicitly solves for the SDC capacity as a function of the complementarity, showing that the entropy of the overlap matrix gives the capacity, when the preshared state is maximally entangled. We generalise this equation to resources with symmetric noise such as a preshared Werner state. In the most general case of arbitrary noisy resources, we obtain an analogous lower bound on the SDC capacity. Our results shed light on the role of complementarity in determining the quantum advantage in SDC and also seem fundamentally interesting since they bear a striking resemblance to uncertainty relations.
\end{abstract}


\maketitle

\section{Introduction}

Quantum physics has revolutionised the study of information theory, as many information-processing tasks can be better accomplished using the non-classical weirdness of quantum systems. The most famous example is the exponential speed-up for certain computing tasks that could (in principle) be obtained for a quantum computer \cite{Steane1998,NieChu00}. Another example is the in-principle perfect security offered by quantum cryptography \cite{QcryptRevModPhys.74.145,QKDRevModPhys.81.1301}. From a fundamental perspective, it is natural to ask: what aspect of quantum systems or quantum dynamics leads to the information-processing advantage? This is the well-known \textit{quantum advantage} question.

An enormous amount of research has gone into studying the role of \textit{quantum correlations}, such as entanglement \cite{HHHH09} and discord \cite{ModiEtAlRevModPhys.84.1655}, in determining the quantum advantage. While quantum correlations are indeed important and useful, they are not necessarily the whole story, since there are non-classical aspects of \textit{unipartite} quantum systems. One such notion is that of \textit{complementarity}, the idea that certain pairs of observables cannot simultaneously be known, famously stated in Heisenberg's uncertainty principle \cite{Heisenberg}. Indeed, the uncertainty principle provided inspiration for the original proposal for quantum key distribution (QKD) \cite{Wiesner}, and recently it was explicitly used to prove the security of QKD \cite{TLGR}.

While it is common to think of entanglement or discord as a resource, it is much less common to think of complementarity as a resource, perhaps because the notion seems somewhat vague. However, complementarity is naturally quantified in uncertainty relations, which are quantitative bounds on one's knowledge of (typically) pairs of observables. One of the most famous uncertainty relations is that of Maassen and Uffink \cite{MaassenUffink}, where they consider two orthonormal bases $X=\{ \ket{x_l} \}$ and $Z=\{\ket{k}\}$ on the Hilbert space $\HC_A$ of quantum system $A$ (w.l.o.g.\ we take $Z$ to be the standard basis). For any state $\rho_A$ the uncertainties, quantified by the Shannon entropies $H(X)$ and $H(Z)$, are bounded by
\begin{equation}
\label{eqn1}
H(X)+H(Z)\geq \log_2 (1/c)
\end{equation}
where
\begin{equation}
\label{eqn2}
c = \max_{k,l} c_{kl},\quad \text{with  }c_{kl}= |\ip{k}{x_l}|^2, 
\end{equation}
quantifies the complementarity between the $X$ and $Z$ observables. The complementarity factor $c$ appears in other uncertainty relations \cite{EURreview1}, including strong versions that have application for QKD security analysis \cite{RenesBoileau, BertaEtAl, TomRen2010, ColesColbeckYuZwolak2012PRL}.

In this article, we investigate the role that complementarity plays in a simple, well-known communication task called superdense coding (SDC) \cite{BennWiesPhysRevLett.69.2881}. In studying the quantum advantage it is natural to seek out the simplest imaginable tasks for which quantum systems provide an advantage over classical systems. Indeed SDC is extremely simple: by sending a \textit{single} qubit to Bob, Alice can convey \textit{two} classical bits of information. In the ideal case, Alice and Bob preshare a maximally entangled state (MES) $\ket{\phi}$ for two qubits, Alice encodes her messages by applying a Pauli operator unitary $\sg_j$ chosen from the set $\{\id, \sg_X, \sg_Y, \sg_Z\}$ to system $A$ resulting in the state $(\sg_j\ot \id)\ket{\phi}$, she sends $A$ to Bob, who can then perfectly distinguish between the four possible states by doing a measurement in the Bell basis. This simple example dramatically illustrates the quantum advantage: with a classical strategy Alice can at best send \textit{one} bit of information for each two-level system that she sends, while the quantum strategy gives \textit{two}.

So what is behind the quantum advantage in SDC? The typical answer is entanglement, and indeed there would be no advantage over a classical strategy if Alice and Bob's preshared state was disentangled \cite{BosePlenVedr2000, BowenPRA.63.022302, HiroshimaJPA2001, BrussEtAlPRL.93.210501}. But entanglement is not the only nonclassical aspect of SDC; there is also the fact that the local unitaries that Alice may implement are non-commuting. In fact, we show that if Alice's set of local unitaries is commutative, then she cannot obtain any advantage over a classical strategy, regardless of whether the preshared state with Bob is entangled. It seems that both entanglement \textit{and} complementarity are important for SDC.

We further attempt to capture the importance of complementarity in SDC by making an analogy to uncertainty relations. Notice that Alice's set of Pauli operator unitaries can be written as $\{ \sg_X^m \sg_Z^n \}$ where both $m$ and $n$ are either 0 or 1, and the phase factor relating $\sg_Y$ to $\sg_X\sg_Z$ has no consequence. Writing the unitaries like this emphasizes that the non-commutativity of the unitaries is linked to the fact that there are two complementary bases, $X$ and $Z$. Suppose we allow the angle between these two bases to vary, possibly decreasing the complementarity of the two bases, how does this affect the capacity of SDC? The extreme case in which the $X$ and $Z$ bases are fully complementary allows her to send two bits, while the other extreme corresponds to $X = Z$ so Alice's unitaries are commuting and hence she can send no more than one bit. For the intermediate cases, we find - in our main result - a fascinating connection between the SDC capacity and uncertainty relations such as \eqref{eqn1}. Alice's classical capacity $\CC(X,Z)$ as a function of the $X$ and $Z$ bases is given by:
\begin{align}
\label{eqn3aaa}\CC(X,Z) &= [1 + H_{\bin}(c)] \text{  bits}\\
\label{eqn3}&\geq [1 + \log_2 (1/c)] \text{  bits},
\end{align}
where $H_{\bin}$ is the binary entropy and the $c$ here is \emph{precisely the same factor appearing in the uncertainty relation \eqref{eqn1}}. Note that the extreme cases of $c = 1/2$ and $c = 1$ respectively give two bits and one bit for our capacity formula. We state the lower bound \eqref{eqn3} here because we can only find a simple equation like \eqref{eqn3aaa} in some special cases of high symmetry, whereas were shall generalise the bound in \eqref{eqn3} to the most general case of noisy resources (discussed below).

Operationally, one may think of the partial complementarity case as a scenario where Alice attempts to do the SDC encoding with her four unitaries but her Hadamard gate $\Hbb$, which allows her to transform from the standard basis to the $X$ basis, is faulty \footnote{Bob may need to know this in order to optimally decode her messages.}, and she implements $\sgt_X = \Hbb \sg_Z \Hbb\ad$ instead of the usual $\sg_X$. This scenario is depicted and discussed in Fig.~\ref{fgr1}. In that case, one can rewrite \eqref{eqn2} as $c = \max_{k,l} |\mted{k}{\Hbb}{l}|^2 $, so $c$ quantifies how close Alice's faulty Hadamard is to an actual Hadamard, and as long as it is fairly close then \eqref{eqn3aaa} guarantees that Alice can send information at a rate close to 2 bits.

Our work is part of a broader effort to understand the importance of complementarity for quantum information processing tasks.  Brandao and Horodecki found a relation between $\log_2 (1/c)$ and the speed-up that quantum computation obtains over classical computation in Ref.~\cite{BranHoro2013}. Coles and Piani \cite{ColesPiani2013arXiv1305.3442C} considered the task of generating entanglement with sequential measurements, which is related to decoupling and coherent teleportation, and likewise obtained bounds involving $\log_2 (1/c)$. Renes gives some nice qualitative intuition about the importance of complementarity for several tasks, including SDC, in Ref.~\cite{Renes2012arXiv1212.2379R} (see references therein for further discussion). 

In what follows, we introduce our notation and discuss the SDC protocol in the next section. We give some intuition in Sec.~\ref{sct2C}, then we briefly remark in Sec.~\ref{sct37} that non-commuting unitaries are crucial for SDC. Section~\ref{sct3UncRel} presents our main results. In general, Alice's and Bob's preshared state is a density operator $\rho_{AB}$, and after applying her unitary to $A$, Alice sends $A$ to Bob over a noisy quantum channel $\EC$.  Section~\ref{sct38abaae123} presents our exact results, i.e., equations that generalise \eqref{eqn3aaa} to arbitrary dimensions and also allow for some special types of resource noise. Sections~\ref{sct38} considers even more general preshared states $\rho_{AB}$ and noisy channels $\EC$, while presenting various generalisations of the lower bound in \eqref{eqn3}.  A summary of our results as well as possible future extensions are discussed in Sec.~\ref{sct45}.

\begin{figure}[t]
\begin{center}
\vspace{2pt}
       
\includegraphics[width=7.1cm]{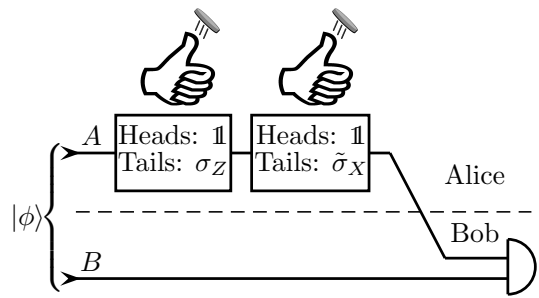}

\caption{The following is the basic idea of our main result, specialised to the case where Alice and Bob preshare a MES $\ket{\phi}$ of two qubits $A$ and $B$. Alice flips a coin to determine if she will apply $\id$ or $\sg_Z$ to $A$, and a second coin to determine if she subsequently applies $\id$ or $\sgt_X$ to $A$. She then sends $A$ to Bob (let us assume over a noiseless channel, although we later allow for a noisy channel). In the ideal case where $X$ and $Z$ are fully complementary, Bob can then measure $AB$ in the Bell basis to perfectly determine the outcome of Alice's two coin flips, i.e., Alice sends Bob two ``coins" of information. At the other extreme, when $X$ and $Z$ have zero complementarity, Alice can at best send only one coin of information per qubit that she sends. When $X$ and $Z$ have partial complementarity, as measured by $c$ (see text for definition of $c$), we find that Alice can send $1+H_{\bin}(c)$ coins of information per qubit sent to Bob, in the limit of infinitely many usages of the protocol. In other words, there is a quantitative connection between the complementarity of $X$ and $Z$ and the number of ``coins" that can be sent using superdense coding. (In general the information transmitted might not be about the individual physical coin flips, but might be about correlations, e.g., whether two flips gave the same or opposite outcomes. Hence we will measure information in the abstract unit of ``bits" rather than ``coins".)
\label{fgr1}}
\end{center}
\end{figure}

\section{Preliminaries}\label{sct2}

\subsection{Superdense coding protocol} \label{sct2B}

The SDC protocol consists of three steps:  (1) Encoding, (2) Transmission, and (3) Decoding. Or depending on one's view there are four steps with the initial one being (0) Presharing correlations. We denote the state preshared by Alice and Bob as $\rho_{AB}$, for quantum systems $A$ and $B$. While for simplicity we take $\ddd=\dim{(\HC_A)}=\dim{(\HC_B)}$, our results naturally generalise to the case of $\dim{(\HC_A)}\neq \dim{(\HC_B)}$. Throughout this article, we adopt the usual notation of $\sg_A=\Tr_B(\sg_{AB})$ and $\sg_B=\Tr_A(\sg_{AB})$ for the marginals of a bipartite state $\sg_{AB}$.

In the encoding step Alice encodes her message by applying a local unitary $U_j$ to system $A$ with probability $p_j$, resulting in the encoded state $\rho'_{AB,j}$, and the ensemble average of the encoded states is denoted $\rho'_{AB}$, i.e.,
\begin{equation}
\label{eqn4}
\rho'_{AB,j}= (U_j \ot \id)\rho_{AB}(U_j\ad \ot \id),\quad \rho'_{AB}=\sum_j p_j \rho'_{AB,j}.
\end{equation}
We note that more general local operations have been considered as encodings in, e.g., \cite{Horodecki:2001:CCN:2011339.2011345, HoroPianJPA2012,ShadEtAlNJP2010,ShadEtAlReview2013}. Another generalisation, assuming Alice and Bob preshare $N$ copies of $\rho_{AB}$, considers global unitaries applied to $A^{\ot N}$, although in several important cases no advantage is gained by applying global versus local unitaries \cite{BrussEtAlPRL.93.210501,ShadEtAlPRA2012}. For our purposes, local unitaries are sufficient for demonstrating the importance of complementarity for SDC, hence we focus on local unitary encodings.

We are particularly interested in sets of unitaries $\{U_j\}$ corresponding to products of Pauli operators, since as we shall see they are often optimal for encoding in SDC. We denote the set of all such $\ddd^2$ Pauli products as $\{\sg_X^m \sg_Z^n\}$, where $m$ and $n$ each run from 0 to $\ddd-1$ and
\begin{equation}
\label{eqnPauliDef1}
\sg_Z = \sum_{k=0}^{\ddd-1} \om^{k}\dya{k},\quad \sg_X = \sum_{k=0}^{\ddd-1} \dyad{k+1}{k},
\end{equation}
with $\om = e^{2\pi i/\ddd}$. More generally, we will allow for imperfection in Alice's Pauli operators and consider the set $\{\sgt_X^m \sg_Z^n\}$, where
\begin{equation}
\label{eqnPauliDef2}
\sgt_X = \sum_{k=0}^{\ddd-1} \om^{k}\dya{x_k}=\Hbb\sg_Z\Hbb\ad.
\end{equation}
Here we allow the imperfect Hadamard $\Hbb$ to be any unitary matrix. We have $\sgt_X= \sg_Z$ when $\Hbb$ is the identity, and $\sgt_X= \sg_X$ when $\Hbb = \sum_{k,k'} (\om^{-kk'}/\sqrt{d})\dyad{k}{k'}$ is the Fourier matrix.

Naturally, after the encoding, Alice sends $A$ to Bob. The channel $\EC$ through which $A$ passes may be noisy, in general given by a completely-positive trace-preserving map. We consider a few examples for $\EC$ in Sec.~\ref{sct3UncRel}. After passing through $\EC$ the states in the ensemble are denoted:
\begin{equation}
\label{eqn5}
\rho''_{AB,j} = (\EC\ot \IC)(\rho'_{AB,j}),\quad\rho''_{AB}=\sum_j p_j \rho''_{AB,j}.
\end{equation}
For future reference, let us also define the state in which $\EC$ acts directly on $\rho_{AB}$,
\begin{equation}
\label{eqn6}
\rhoh_{AB} = (\EC\ot \IC)(\rho_{AB}),
\end{equation}
which does not represent a state physically associated with the SDC protocol, but nonetheless appears in the mathematics in Sec.~\ref{sct3UncRel}.

Finally, Bob decodes Alice's messages by measuring the joint $AB$ system, or more generally he may do a global measurement on $A^{\ot N}B^{\ot N}$ assuming Alice and Bob repeated the first steps of SDC $N$ times. We place no restrictions on Bob's decoding.

\subsection{Superdense coding capacity} \label{sct2D}

There are several entropic measures that often show up in expressions for the classical capacity; let us define these measures. The Shannon entropy, von Neumann entropy, conditional von Neumann entropy, and relative entropy are respectively defined as 
\begin{align}
\label{eqnEntDef1}
H(\{p_j\})&= -\sum_j p_j \log p_j,\notag \\
S(\rho)&= -\Tr(\rho\log \rho),\notag \\
S(A|B)_{\sg}& = S(\sg_{AB}) - S(\sg_B), \notag \\
 D(\rho || \sg) &= \Tr(\rho \log \rho)-\Tr(\rho \log \sg),
\end{align}
where all logarithms are base 2 in this article (henceforth we drop the subscript). The binary entropy is defined as $H_{\bin}(p) = H(\{p,1-p\})$. Also, a crucial quantity for classical capacity is Holevo's $\chi$ quantity, defined for an ensemble $\{p_j, \rho_j\}$ as \cite{NieChu00}:
\begin{equation}
\label{eqnchiDef1}
\chi(\{p_j,\rho_j\}) = S \Big(\sum_j p_j \rho_j\Big) - \sum_j p_j S(\rho_j).
\end{equation}

The classical capacity is often evaluated in the asymptotic limit that Alice and Bob have $N$ copies of their systems, where $N\to \infty$, and the capacity is then defined as the optimal rate (bits per copy) that Alice can reliably send to Bob \cite{NieChu00}. In SDC, each copy received by Bob corresponds to the channel $\EC$ acting on $A$ after Alice encoded with a unitary $U_j$ on $A$ when the preshared state is $\rho_{AB}$, resulting in a state $\rho''_{AB,j}$ from \eqref{eqn5}. For a given set of unitaries $\{ U_j\}$, we naturally allow Alice to optimize over all possible probability distributions $\{p_j\}$ for these unitaries. (We assume she uses the same unitary set $\{U_j\}$ and probabilities $\{p_j \}$ for each $A$ copy.) Then the classical capacity is given by \cite{Schu97}: 
\begin{align}
\CC(\{U_j\},\rho_{AB},\EC) & = \max_{ \{p_j\} } \chi ( \{p_j, \rho''_{AB,j} \} )\notag \\
&= \max_{ \{p_j\} } [S(\rho''_{AB})-\sum_j p_j S(\rho''_{AB,j})].
\label{eqn7}
\end{align}

In the case where Alice's unitaries are unrestricted, it is natural to maximise over all possible unitary sets:
\begin{equation}
\label{eqn235CapDef}
\CC(\rho_{AB},\EC)=  \max_{ \{U_j\} } \CC(\{U_j\},\rho_{AB},\EC) 
\end{equation}
and this is often called the superdense coding capacity (for a given preshared state and noisy channel).

\section{Intuition for the quantum advantage} \label{sct2C}

The obvious strategy for Alice to send Bob classical information is for her to physically send classical bits or more generally classical dits. Since the world is quantum mechanical (or so we assume), a classical system is a particular kind of quantum system, namely a decohered one \cite{ZurekReview}. So we can speak of the different density operators $\rho_j$ of a $d$-sided classical die, each associated with a different outcome for rolling the die, where outcome $j$ has probability $p_j$. Let us imagine that Alice rolls the die and then sends it to Bob in a manner that preserves the roll outcome. So with probability $p_j$, Bob receives density operator $\rho_j$. The Holevo quantity for this protocol satisfies \cite{NieChu00}
\begin{equation}
\label{eqn52804}
\chi ( \{p_j, \rho_j \} )\leq H( \{ p_j \}) \leq \log d.
\end{equation}
More generally, we could allow the die to pass through a noisy channel on the way to Bob, but we end with the same result that $\log d$ upper bounds the Holevo quantity.  We shall refer to this as the ``classical strategy". Because the Holevo quantity in \eqref{eqn52804}, optimized over all $\{p_j\}$, represents the classical capacity for this strategy, then the capacity is upper bounded by $\log d$. In what follows, we will say that one obtains a quantum advantage if the superdense coding capacity exceeds the $\log d$ classical bound.

Eq.~\eqref{eqn52804} shows that the Holevo quantity is restricted by the size of the sample space. Hence we can now see the intuition behind superdense coding. The use of entanglement together with non-commuting unitary encodings allows Alice to effectively expand the size of the sample space, from $d$ to $d^2$. Or let us state it in the opposite way. If Alice tried to implement SDC in a classical (i.e., high decoherence) setting, then decoherence would effectively reduce the size of the sample space, from $d^2$ down to $d$, such that \eqref{eqn52804} holds. Decoherence destroys the preshared entanglement and it effectively turns the non-commuting unitaries into commuting ones. One can see the latter by taking $Z$ to be the pointer basis (the basis in which state is diagonal \cite{ZurekReview}), then $\sg_Z$ commutes with the state, $\sg_Z\rho\sg_Z\ad = \rho$, and since $\sg_X$ performs a shift in the $Z$ basis, we have $\sg_X\sg_Z \rho \sg_Z\ad\sg_X\ad = \sg_Z \sg_X\rho \sg_X\ad\sg_Z\ad$. The result is that distinct encodings become effectively equivalent. For example for Pauli product encodings we effectively get $\sg_Z = \id$ and $\sg_X\sg_Z = \sg_X$, which is what we mean by effectively contracting the sample space (from 4 to only 2 distinct encodings in the qubit case). Indeed, in the next section, we show that commuting unitaries kill any hopes to exceed the $\log d$ bound using SDC.

\section{Commuting unitaries}\label{sct37}

Here we briefly remark on the necessity of non-commuting unitaries in order to obtain an advantage over the classical strategy, i.e., where Alice simply sends Bob a classical dit with $\ddd$ possible states. Since the idea of two bases being complementary is connected to their associated Pauli operators being non-commutative, we use the notions of complementarity and non-commutativity somewhat interchangeably \footnote{After all, the original uncertainty relation from Robertson quantified complementarity using the expectation value of the commutator of the Pauli operators \cite{Robertson} (whereas in modern times we have switched to using the overlaps $c_{kl}$ to quantify complementarity).}.  Language aside, we have the following result, proved in Appendix~\ref{app0}.
\begin{proposition}
\label{prop13454}
Let $\rho_{AB}$ be an arbitrary quantum state, let $\EC$ be an arbitrary quantum channel, and let $\{ U_j \}$ be any set of commuting unitaries, then  
\begin{equation}
\label{eqn8}
\CC(\{U_j\}, \rho_{AB},\EC) \leq \log \ddd. 
\end{equation}
\end{proposition}

The bound in \eqref{eqn8} shows that, no matter how entangled $\rho_{AB}$ is, the SDC protocol offers no advantage over the classical strategy when Alice's unitaries are commuting. In contrast to \eqref{eqn8}, when Alice's unitaries are unrestricted, she can achieve a rate of $2 \log \ddd$ when $\rho_{AB}$ is maximally entangled and $\EC$ is noiseless.

We note that Prop.~\ref{prop13454} lends itself to \textit{witnessing complementarity} using SDC in the following sense. If Alice has a black box that implements some unitary randomly chosen from a set of unitaries, then she can witness the non-commutativity of the unitary set by verifying that she can achieve a quantum advantage (i.e., exceed the classical bound $\log \ddd$) using the box.

\section{Quantitative dependence on complementarity}\label{sct3UncRel}

\subsection{Introduction}\label{sct37.9}

We now present our main results, which give quantitative connections between complementarity and SDC capacity. To formulate them, we suppose that Alice encodes her messages with the Pauli product unitaries $\{ \sgt_X^m\sg_Z^n \}$ (see Sec.~\ref{sct2B}). We then quantify the complementarity of these encodings by considering the ``overlap matrix" $c_{kl}$ defined in \eqref{eqn2}. In the qubit case ($d=2$), a single parameter $c=\max_{k,l} c_{kl}$ characterises the whole matrix, and therefore we obtain a simple dependence of the SDC capacity on $c$. For $d>2$, we find that the SDC capacity depends on the entire matrix $c_{kl}$. Nevertheless we are able to state the exact dependence on $c_{kl}$ for any $d$, so long as the preshared state $\rho_{AB}$ is a Werner state and $\EC$ is a depolarising channel. We state this exact result in Sec.~\ref{sct38abaae123}.

For more general types of noisy resources, $\rho_{AB}$ and $\EC$, we state our results as lower bounds on the SDC capacity in terms of $\log(1/c)$. As noted in the Introduction, this is same factor appearing in the Maassen-Uffink uncertainty relation. Section~\ref{sct382z} presents our lower bounds for some special resource choices, while Sec.~\ref{sct42} generalises these bounds to the most general case of arbitrary $\rho_{AB}$ and $\EC$. Finally Sec.~\ref{sct3345} notes that one can obtain stronger bounds (for $d>2$) if one considers more elements of the matrix $c_{kl}$ than just $c$.

While Prop.~\ref{prop13454} showed that complementarity is \textit{necessary} to achieve a quantum advantage with SDC, the following results allow us to also argue that complementarity is \textit{sufficient} for a quantum advantage. This is especially true in the case of a maximally entangled preshared state, which we discuss now, where the relationship between complementarity and the quantum advantage is very clear and precise.

\subsection{Equations for special cases}\label{sct38abaae123}

\subsubsection{MES, noiseless channel}\label{sct38abaae1}

Let us consider the special case where Alice and Bob preshare a MES $\ket{\phi}$ and the channel $\EC$ is noiseless, denoted as $\EC = \IC$. In this case, we have a simple and precise dependence of the SDC capacity on the elements of the overlap matrix $c_{kl}$, as follows.

\begin{theorem}
\label{thm56aabas}
Let $\ket{\phi}\in \HC_{AB}$ be any MES, then
\begin{align}
\label{eqn15MES} \CC(\{\sgt_X^m \sg_Z^n\},\dya{\phi},\IC)= H(\{c_{kl}/d\}) 
\end{align}
where $c_{kl}= |\ip{k}{x_l}|^2= |\mted{k}{\Hbb}{l}|^2$.
\end{theorem}
We will prove \eqref{eqn15MES} below when we consider a more general case in Thm.~\ref{thmWerDep1}. The formula in \eqref{eqn15MES} is beautifully simple, stating that the SDC capacity is simply given by the entropy of the (normalised) overlap matrix. It is interesting that every matrix element $c_{kl}$ enters on equal footing into the formula in \eqref{eqn15MES}.

We can rewrite \eqref{eqn15MES} in a way that clearly shows the term contributing to the quantum advantage, as follows:

\begin{align}
\label{eqn15MESa} 
\CC(\{\sgt_X^m \sg_Z^n\},\dya{\phi},\IC)= \log d+ \sum_{k} (1/d) H( \{c_{kl}\}_l  )
\end{align}
where $H( \{c_{kl}\}_l  ) = H(X)_{\ket{k}}$ is the entropy of row $k$ of the matrix $c_{kl}$ \footnote{For clarity we may add a subscript on the set, as in $\{c_{kl}\}_l$, to indicate the running index.}. Equation~\eqref{eqn15MESa} says that the quantum advantage is given by the average entropy of the rows of $c_{kl}$ (or equivalently the average over the columns). This equation directly implies the following observation.

\begin{proposition}
\label{thm56aaswbas}
Let $\ket{\phi}\in \HC_{AB}$ be any MES. Then
\begin{align}
\label{eqn15MESpropa} \CC(\{\sgt_X^m \sg_Z^n\},\dya{\phi},\IC)= 2\log d 
\end{align}
if and only if $\Hbb$ is a Hadamard matrix, i.e., of the form
\begin{equation}
\label{eqn4928}
\Hbb =\sum_{k,l} (e^{i \phi_{k,l}}/\sqrt{d})\dyad{k}{l}
\end{equation}
for some phase factors $e^{i \phi_{k,l}}$ chosen so that $\Hbb$ is unitary. At the other extreme, we have
\begin{align}
\label{eqn15MESpropb} \CC(\{\sgt_X^m \sg_Z^n\},\dya{\phi},\IC)= \log d 
\end{align}
if and only if $\Hbb$ is a trivial relabelling of the standard basis with phases applied, i.e., of the form 
\begin{equation}
\label{eqn4929}
 \Hbb = \sum_k e^{i\phi_k}\dyad{P(k)}{k}
\end{equation}
where the $e^{i\phi_k}$ are arbitrary phase factors and $P()$ is a permutation function. 
\end{proposition}

This proposition provides the precise statement that allows us to argue that complementarity is both necessary and sufficient to achieve a quantum advantage in SDC. It is natural to say that the $\Hbb$ in \eqref{eqn4928} and \eqref{eqn4929} corresponds to the cases of full complementarity and zero complementarity, respectively. Hence, Prop.~\ref{thm56aaswbas} says that the SDC capacity is 1 dit if and only if the encodings have zero complementarity, and it is 2 dits if and only if the encodings have full complementarity.

Specialising \eqref{eqn15MESa} to the case of $d=2$ gives
\begin{align}
\label{eqn15qubit} \CC(\{\sgt_X^m \sg_Z^n\},\dya{\phi},\IC)= 1+H_{\bin}(c).
\end{align}
Hence the SDC capacity depends only on a single parameter $c$; we plot this dependence in Fig.~\ref{fgr2} (the Noise~$=0$ curve). We can see that this plot is consistent with Prop.~\ref{thm56aaswbas}, i.e., the capacity is 2 bits iff $c=1/2$ and 1 bit iff $c=1$.   

We show below (Lemma~\ref{thmEqProbs}) that the optimal strategy - to achieve the rate in \eqref{eqn15qubit} - is for Alice to encode with equal probabilities ($p_j = 1/4, \forall j$), regardless of the value of $c$. We can think of this encoding strategy as basically two sequential coin flips, and hence this justifies our depiction in Fig.~\ref{fgr1} of the Introduction.

\subsubsection{Werner state, depolarising channel}\label{sct38abaae2}

\begin{figure}[t]
\begin{center}
\includegraphics[width=8.2cm]{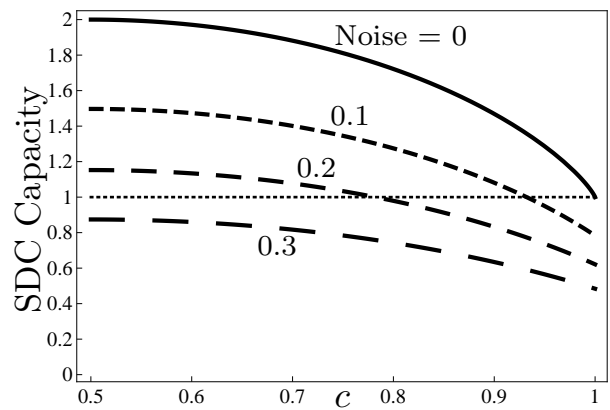}
\caption{Plot of $\CC(\{\sgt_X^m \sg_Z^n\},\rho_{AB}^{\al},\EC_d^{\bt})$ as a function of $c$ for $d=2$. Different curves correspond to different values of noise, given by $(1-\al\bt)$. The region above the horizontal dotted line (capacity equal to 1 bit) corresponds to a quantum advantage.
\label{fgr2}}
\end{center}
\end{figure}

The above result generalises to a noisy scenario where Alice and Bob preshare a Werner state and $A$ is transmitted over a depolarising channel. Werner states are a weighted mixture of a MES $\dya{\phi}$ with the maximally mixed state, of the form:
\begin{equation}
\label{eqn24}
\rho^{\al}_{AB} = \al \dya{\phi} +(1-\al) \id / \ddd^2
\end{equation}
where $0\leq \al \leq 1$. Similarly the action of a depolarising channel \cite{NieChu00} on a normalised density operator is:
\begin{equation}
\label{eqn16}
\EC^{\bt}_d(\rho_A) = \bt \rho_A +(1-\bt )\id/\ddd,
\end{equation}
which interpolates between the noiseless channel and the completely noisy channel as $\bt$ goes from 1 to 0. Notice that $\rho^{\al}_{AB} = (\EC^{\al}_d\ot \IC)(\dya{\phi})$. Also notice that the composition of two depolarising channels is a more noisy depolarising channel: $\EC^{\bt}_d\circ \EC^{\al}_d = \EC^{\al \bt}_d$. Thus, for a depolarising channel acting on a Werner state, $(\EC^{\bt}_d\ot \IC) (\rho^{\al}_{AB})=(\EC^{\al\bt}_d \ot \IC)(\dya{\phi})$ and the noise enters in as a single parameter: $(1-\al\bt)$.

The following lemma shows that we only need to consider the equal probability case when evaluating the capacity for these states and channels, with the proof given in Appendix~\ref{app1} \footnote{A key idea for the proof of Lemma~\ref{thmEqProbs} was kindly provided by J. Kaniewski.}.
\begin{lemma}
\label{thmEqProbs}
Let $\rho^{\al}_{AB}$ be a Werner state as defined in \eqref{eqn24}, let $\EC^{\bt}_d$ be defined by \eqref{eqn16}, then
\begin{align}
\label{eqnWerDepEqProb} \CC(\{\sgt_X^m \sg_Z^n\}, \rho^{\al}_{AB}, \EC^{\bt}_d)= \chi ( \{1/d^2, \rho''_{AB,j} \} ). 
\end{align}
In other words, choosing equal probabilities $p_j = 1/d^2, \forall j,$ is optimal.
\end{lemma}

Now let us generalise Thm.~\ref{thm56aabas} to a noisy case of Werner states and depolarising channels, with the proof in Appendix~\ref{app2}.

\begin{theorem}
\label{thmWerDep1}
Let $\rho^{\al}_{AB}$ be a Werner state as defined in \eqref{eqn24}, let $\EC^{\bt}_d$ be defined by \eqref{eqn16}, then
\begin{align}
\label{eqn15WerDep} \CC(\{\sgt_X^m \sg_Z^n\}, \rho^{\al}_{AB}, \EC^{\bt}_d)&= H(\{ \frac{\al\bt c_{kl}}{d} +\frac{1-\al\bt }{d^2} \}_{kl} )\notag\\
&\hspace{10pt} - S(\rhoh_{AB}) , 
\end{align}
where $\rhoh_{AB} = (\EC^{\bt}_d \ot \IC)(\rho^{\al}_{AB})$. 
\end{theorem}

Figure~\ref{fgr2} plots the SDC capacity for various values of the noise $(1-\al\bt)$. The capacity monotonically decreases with $c$. Notice that the dependence on $c$ is less dramatic as the noise increases, and likewise the dependence on noise is less dramatic as $c$ increases. This suggests that entanglement and complementarity ``work together"; their effects are non-independent. In other words, the presence of one resource gives the potential to unlock the full power of the other resource to raise the capacity.

When the capacity is above the horizontal dotted line in Fig.~\ref{fgr2}, there is a quantum advantage; so clearly low $c$ and low noise are desirable to be in this region of the plot. For certain values of the noise, namely $(1-\al\bt)\geq 0.253$, there is no quantum advantage for all $c$.

\subsection{Lower bounds in terms of complementarity}\label{sct38}

\subsubsection{More special cases}\label{sct382z}

The symmetric nature of the resources considered above make equal probabilities the optimal encoding strategy and hence allow us to remove the optimisation over probability distributions when evaluating the capacity. We now consider more general resources. Although we cannot explicitly solve for the dependence of the capacity on complementarity in the general case, we can lower-bound the capacity in terms of the factor $\log(1/c)$. For example, it is easy to verify $H_{\bin}(c)\geq \log(1/c)$, hence \eqref{eqn15qubit} implies the lower bound $1+ \log(1/c)$ stated in \eqref{eqn3}. More generally, for arbitrary $d$, \eqref{eqn15MESa} implies that $\log d + \log(1/c)$ lower-bounds the SDC capacity, as can be easily verified, when the preshared state is a MES $\ket{\phi}$ and $\EC$ is noiseless.

Now we generalise this bound one step further, to the case where the preshared state $\rho_{AB}$ is arbitrary (but the channel $\EC$ is still noiseless), with the proof in App.~\ref{app3aaa}. In addition to stating our lower bound in terms of $\log(1/c)$, for completeness we also note that using full complementarity ($c=1/d$) achieves the overall capacity, \eqref{eqn10}, which is a well-known formula \cite{BowenPRA.63.022302, Ban2002, HiroshimaJPA2001, BrussEtAlPRL.93.210501}. Notice that \eqref{eqn10} gives the familiar $2\log \ddd$ result in the case of maximal entanglement, where $-S(A|B)_{\rho}=S(\rho_B)-S(\rho_{AB}) = \log \ddd$. More generally, $-S(A|B)_{\rho}$ quantifies the entanglement preshared between Alice and Bob \cite{DevWin03}.

\begin{proposition}
\label{thm56}
For any state $\rho_{AB}$, the overall SDC capacity, as shown in Refs.~\cite{BowenPRA.63.022302, Ban2002, HiroshimaJPA2001, BrussEtAlPRL.93.210501}, is
\begin{equation}
\label{eqn10}
 \CC(\rho_{AB},\IC) = \log \ddd - S(A|B)_{\rho} \, ,
\end{equation}
while for the encodings $\{\sgt_X^m \sg_Z^n\}$ we have
\begin{align}
\label{eqn15} \CC(\{\sgt_X^m \sg_Z^n\},\rho_{AB},\IC)\geq \log (1/c) -S(A|B)_{\rho} \,  .
\end{align}
\end{proposition}

The bound in \eqref{eqn15} has a very simple structure. It captures the contributions of \textit{both} of Alice's resources for SDC. While the term $-S(A|B)_{\rho}$ quantifies her preshared entanglement, the term $\log (1/c) $ is the same one appearing in the uncertainty relation of Maasen and Uffink, \eqref{eqn1}, and hence quantifies her complementarity resource. As long as the sum of the two terms does not dip below $\log \ddd$, then \eqref{eqn15} guarantees that SDC is advantageous over a classical strategy for sending information. Despite the simple additive nature of this bound, where it seems like entanglement and complementarity contribute independently to the capacity, note that our previous discussion of Fig.~\ref{fgr2} suggests otherwise. It is likely that the two resources ``work together" (i.e., in a non-independent way) to raise the capacity.

The above statements for the noiseless channel naturally generalise to the depolarising channel, as follows.
\begin{proposition}
\label{thm56456}
Let $\rho_{AB}$ be any state, let $\EC^{\bt}_d$ be the depolarising channel in \eqref{eqn16}, then
\begin{equation}
\label{eqn17}
 \CC(\rho_{AB}, \EC^{\bt}_d) = \log \ddd - S(A|B)_{\rhoh} \, ,
\end{equation}
where $\rhoh_{AB}= (\EC^{\bt}_d \ot \IC)(\rho_{AB})$.
In addition,
\begin{align}
\label{eqn18} \CC(\{\sgt_X^m \sg_Z^n\},\rho_{AB},\EC^{\bt}_d)\geq \log (1/c) -S(A|B)_{\rhoh} \, .
\end{align}
\end{proposition}
The proof of Prop.~\ref{thm56456} is essentially identical to that of Prop.~\ref{thm56} given in App.~\ref{app3aaa}, provided one makes the observation \cite{ShadEtAlNJP2010} that $S(\rho''_{AB,j}) = S(\rhoh_{AB})$ because the action of $\EC^{\bt}_d$ commutes with the action of the unitary $U_j$. Again, $\log(1/c)$ in \eqref{eqn18} quantifies Alice's complementarity resource, but notice now that the contribution from the entanglement resource is a function of the post-channel state $\rhoh_{AB} $.

Propositions~\ref{thm56} and \ref{thm56456} considered an arbitrary preshared state paired with a special noisy channel. Now let us consider the opposite case, where the noisy channel is arbitrary but the preshared state has a special form. A slightly more general case than that of Thm.~\ref{thmWerDep1} considers a preshared Werner state $\rho^{\al}_{AB}$ and transmission over an arbitrary channel $\EC$. The structure of the following result is similar to that of Props.~\ref{thm56} and \ref{thm56456}. The proof is in App.~\ref{app3}.

\begin{proposition}
\label{prop5}
Let $\EC$ be an arbitrary quantum channel and let $\rho^{\al}_{AB}$ be a Werner state defined in \eqref{eqn24}, then \footnote{Ref.~\cite{BanEtAlJPA2004} calculated the overall capacity when the preshared state is maximally entangled and $\EC$ is arbitrary.}
\begin{equation}
\label{eqn25}
 \CC(\rho^{\al}_{AB},\EC) = \log \ddd - S(B|A)_{\rhoh} \, ,
\end{equation}
where $\rhoh_{AB} = (\EC \ot \IC)(\rho^{\al}_{AB})$. In addition,
\begin{align}
\label{eqn26} \CC(\{\sgt_X^m \sg_Z^n\},\rho^{\al}_{AB},\EC)\geq \log (1/c) -S(B|A)_{\rhoh} \, .
\end{align}
\end{proposition}

\subsubsection{The general case}\label{sct42}

In the general case of an arbitrary preshared state and an arbitrary noisy channel, a simple expression for the SDC capacity is not known, unlike some of the previously considered special cases. Nevertheless, we can still give a lower bound on the SDC capacity that depends on the complementarity of Alice's unitary encodings, with the proof in App.~\ref{app4}. Naturally, the general bound is more complicated than that for the special cases considered above.

\begin{theorem}
\label{thm58}
Let $\EC$ be an arbitrary quantum channel and let $\rho_{AB}$ be an arbitrary quantum state, then
\begin{align}
\label{eqn29} &\CC(\{\sgt_X^m \sg_Z^n\},\rho_{AB},\EC)\geq \log( 1/c)+D(\sum_j \frac{1}{\ddd^2}\rho''_{A,j} || \EC(\id))\notag\\
&\hspace{10pt}+S(\rho_B)+S(\sum_j \frac{1}{\ddd^2}\rho''_{A,j}) -  \sum_j \frac{1}{\ddd^2}S(\rho''_{AB,j}). 
\end{align}
\end{theorem}

While this bound is somewhat complicated, it does imply the previous lower bounds given in Props.~\ref{thm56}, \ref{thm56456}, and \ref{prop5}. To appreciate this, consider the special case where $\EC$ is a unital channel, defined by the property $\EC(\id) = \id$. In this case, the second and fourth terms in \eqref{eqn29} cancel, and \eqref{eqn29} reduces to
\begin{equation}
\label{eqn19}
\CC(\{\sgt_X^m \sg_Z^n\},\rho_{AB},\EC) \geq \log (1/c) +S(\rho_B) - \sum_j \frac{1}{\ddd^2}S(\rho''_{AB,j}).
\end{equation}
Noting that the depolarising channel is unital, \eqref{eqn19} further reduces to \eqref{eqn18} by setting $S(\rho''_{AB,j}) =  S(\rhoh_{AB})$ for all~$j$. Alternatively, consider the case where $\EC$ is arbitrary but suppose $\rho_{AB}$ has maximally mixed marginals, $\rho_A = \id /\ddd$ and $\rho_B = \id /\ddd$. This makes $\rho''_{A,j} = \EC(\id / d)= \rhoh_A$ for each $j$ and hence \eqref{eqn29} reduces to
\begin{equation}
\label{eqn27}
\CC(\{\sgt_X^m \sg_Z^n\},\rho_{AB},\EC) \geq \log(1/c) +  S(\rhoh_A) -  \sum_j \frac{1}{\ddd^2}S(\rho''_{AB,j}).
\end{equation}
Noting that Werner states have maximally mixed marginals, \eqref{eqn27} reduces to \eqref{eqn26} also by setting $S(\rho''_{AB,j}) =  S(\rhoh_{AB})$ for all~$j$.

Via the $\log(1/c)$ term in \eqref{eqn29}, we see the importance of complementarity even in the most general case of arbitrary $\EC$ and arbitrary $\rho_{AB}$. One can rewrite the sum of the first two terms in \eqref{eqn29} as $q:= D(\sum_j \frac{1}{\ddd^2}\rho''_{A,j} || \EC(\id/\ddd))-\log( \ddd \cdot c)$ and one can show that $q\leq 0$, so as $c$ approaches $1/\ddd$ (corresponding to optimal complementarity) then $q$ must vanish. Thus, when $c=1/d$, \eqref{eqn29} becomes
\begin{align}
\label{eqn29nbb}
\CC(\{\sgt_X^m \sg_Z^n\},\rho_{AB},\EC)& \geq S(\rho_B) +S(\EC(\id /d)) \notag\\
&\hspace{10pt}- \sum_j \frac{1}{\ddd^2}S(\rho''_{AB,j}),
\end{align}
where the right-hand-side corresponds to the capacity in the case of choosing equal probabilities for the standard Pauli product operators. So the bound in \eqref{eqn29} approaches that in \eqref{eqn29nbb} as the complementarity grows.

\subsubsection{Stronger bound}\label{sct3345}

Thusfar we have used the parameter $c$ in our bounds, mainly because it appears in a well-known uncertainty relation \eqref{eqn1} and may be familiar to readers. However, very recently the bound in \eqref{eqn1} has been strengthened, for $d>2$, in Ref.~\cite{ColesPianiIEUR2013}. Likewise, we can use the technique of \cite{ColesPianiIEUR2013} to strengthen (for $d>2$) our bound on the SDC capacity. This involves replacing Eq.~\eqref{eqn13} from Appendix~\ref{app3aaa} with a stronger statement:
\begin{equation}
\label{eqnBetterbound2343}
\rho'_{AB} \leq \sum_l ( \max_k c_{kl}) \dya{x_l} \ot  \rho_B.
\end{equation}
Using \eqref{eqnBetterbound2343} in place of \eqref{eqn13}, for example in the derivation of \eqref{eqn15}, gives the result
\begin{equation}
\label{eqnBetterbound4423}
\CC(\{\sgt_X^m \sg_Z^n\},\rho_{AB},\IC)\geq f(\rho_A) -S(A|B)_{\rho}
\end{equation}
where 
\begin{equation}
\label{eqnBetterbound4423b}
f(\rho_A):= \sum_l \Big(\sum_k c_{kl} \mte{k}{\rho_A}\Big) \log (1/\max_{k'} c_{k'l} ).
\end{equation}
Note that $f(\rho_A)\geq \log(1/c)$ since averaging over $l$ is larger than minimizing over $l$, so \eqref{eqnBetterbound4423} implies \eqref{eqn15}.

\section{Conclusions}\label{sct45}

The results presented here shed light on the quantum advantage in superdense coding. While it was well-known that entanglement is crucial for SDC, we showed in Prop.~\ref{prop13454} that \textit{complementarity} is also crucial for SDC. Without non-commuting unitaries, Alice can communicate no more than 1 dit of information for each qudit that she sends, even if she has preshared entanglement with Bob. In this sense, one could experimentally use SDC as a complementarity witness: if Alice can communicate more than 1 dit per qudit sent, then her encodings exhibit complementarity.

We then presented quantitative relations between the amount of complementarity and the SDC capacity, in Sec.~\ref{sct3UncRel}. To formulate these relations, we supposed that Alice uses Pauli products to encode her messages, but her Pauli operators may be ``imperfect", i.e., the unitary gate $\Hbb$ that relates the $X$ Pauli operator to the $Z$ Pauli operator may not be a perfect Hadamard. We quantified the complementarity of Alice's encodings by considering the overlap matrix $c_{kl}= |\ip{k}{x_l}|^2= |\mted{k}{\Hbb}{l}|^2$.

Our simplest and perhaps most beautiful result was the exact expression for the SDC capacity, given by $H(\{c_{kl}/d\})$, when Alice and Bob preshare a MES and communicate over of a noiseless channel. Thus we found that the SDC capacity depends on every single element of the matrix $c_{kl}$, is symmetric under interchanging any matrix elements of $c_{kl}$, and is given by simply taking the Shannon entropy of this matrix. This formula for the SDC capacity allowed us to argue, in Prop.~\ref{thm56aaswbas}, that the capacity is 2 dits iff the encodings have fully complementarity, and is 1 dit iff the encodings have zero complementarity.

We then extended this result to the case where the resources have noise, albeit a symmetric sort of noise. Allowing the preshared state to be of the Werner type and the channel to be depolarising, we explicitly solved for the SDC capacity as a function of the overlap matrix in Thm.~\ref{thmWerDep1}. An important observation in proving this was to show that encoding with equal probabilities is the optimal strategy in this case. The curves in Fig.~\ref{fgr2} for different values of noise showed that the dependence on complementarity is less dramatic when the preshared entanglement is lower, \emph{and vice-versa}, suggesting that complementarity and entanglement ``work together" to raise the SDC capacity. In Sec.~\ref{sct38} we considered more general resources, and although we could not give an equation relating capacity to complementarity in the general case, we obtained lower bounds on the capacity in terms of $\log(1/c)$. Theorem~\ref{thm58} gave such a bound in the most general case of arbitrary resources.

The overlap matrix appearing in our results is the same one appearing in well-known entropic uncertainty relations. In particular $c$ appears in the Maassen-Uffink uncertainty relation \eqref{eqn1} and others \cite{EURreview1, RenesBoileau, BertaEtAl, TomRen2010, ColesColbeckYuZwolak2012PRL}. Thus, our equations and bounds for the SDC capacity resemble uncertainty relations. (Since many of our relations take into account more elements of the matrix $c_{kl}$, they \textit{especially} resemble information exclusion relations \cite{ColesPianiIEUR2013}.) From the perspective of fundamental physics, it seems fascinating that the notion of complementarity, that observables come in pairs that exhibit a knowledge tradeoff, can be stated in a completely different way from Heisenberg's original view \cite{Heisenberg}. Our results indicate that \textit{the complementarity of two observables is quantitatively captured by their usefulness for superdense coding}. This is a conceptually new paradigm for capturing the notion of complementarity.

Now let us consider ways in which our results could be extended. While we have studied the capacity in the asymptotic limit of infinite usages, it should be possible to derive bounds for one-shot dense coding similar to the ones we have given here. This is because the one-shot classical capacity is related to a relative entropy function \cite{WangRennerPRL108.200501} with some properties similar to von Neumann relative entropy \cite{ColesColbeckYuZwolak2012PRL}, namely the data-processing inequality.

We have focused on complementarity used by Alice in the encoding, but it could be interesting to consider Bob's decoding. In the ideal case considered in the Introduction where Alice and Bob preshare a MES of two qubits, Bob measures in the Bell basis to decode Alice's message. A Bell basis measurement can be viewed as requiring complementarity, so it could be interesting to explore restrictions on Bob's decoding.

Finally, our results seem to suggest that complementarity is a resource, consistent with Refs.~\cite{BranHoro2013, ColesPiani2013arXiv1305.3442C, Renes2012arXiv1212.2379R}, and also with other references \cite{Stahlke2013arXiv1305.2186S,BraunGeorgPhysRevA.73.022314} that discussed a possibly related resource called interference \footnote{The interference measures in these references are not obviously related to our complementarity measures.}. This idea remains somewhat vague when one considers, for example, how well developed the resource theory is for entanglement \cite{HHHH09}. Is it possible to formulate a resource theory for complementarity? Perhaps the idea would be that Hadamard gates are the resource, as suggested in Ref.~\cite{BraunGeorgPhysRevA.73.022314}, and one could have a notion of a partial Hadamard analogous to the notion of partial entanglement. Ultimately this could lead us closer to answering the widely debated question: what gives the quantum advantage? 

\section{Acknowledgments}

The author thanks J\k{e}drzej Kaniewski and Dan Stahlke for helpful suggestions. The author is funded by the Ministry of Education (MOE) and National Research Foundation Singapore, as well as MOE Tier 3 Grant ``Random numbers from quantum processes" (MOE2012-T3-1-009).

\appendix

\section{Proof of Prop.~\ref{prop13454}} \label{app0}

\begin{proof}
The fact that the $U_j$ are commuting implies that they can be diagonalised in the same basis, and w.l.o.g.\ let us set this to the standard basis $\{\ket{k}\}$. So let us write $U_j = \sum_k e^{i \phi_{j,k}}\dya{k}$ where the $e^{i \phi_{j,k}}$ are some arbitrary phase factors. Define the quantum channel $\ZC$ that decoheres system $A$ in the standard basis: $\ZC(\rho_A) = \sum_k \dya{k}\rho_A \dya{k} $. Then we have 
\begin{widetext}
\begin{equation}
\label{eqn9}
(\ZC \ot \IC)(\rho'_{AB,j}) = \sum_{k,k',k''} e^{i(\phi_{j,k'}-\phi_{j,k''})}(\dya{k}\ot \id)(\dya{k'}\ot \id)\rho_{AB}(\dya{k''}\ot \id) (\dya{k}\ot \id) =  (\ZC \ot \IC)(\rho_{AB}),
\end{equation}
\end{widetext}
where $(\ZC \ot \IC)(\rho_{AB})=\sum_k q_k \dya{k}\ot \rho_{B,k}$ is a classical-quantum state with $q_k \rho_{B,k}= \Tr_A [(\dya{k}\ot \id)\rho_{AB}]$ and $\Tr ( \rho_{B,k}) =1, \forall k$. Hence \eqref{eqn7} becomes:
\begin{align}
\label{eqn23793563}
\CC(\{U_j\},\rho_{AB},\EC) &\leq \CC(\{U_j\},\rho_{AB},\IC)\notag\\
& = \max_{ \{p_j\} } S(\rho'_{AB})-S(\rho_{AB})\notag\\
&\leq S((\ZC \ot \IC)(\rho_{AB})) - S(\rho_{AB})\notag\\
&= S((\ZC \ot \IC)(\rho_{AC})) - S(\rho_{C})\notag\\
&\leq S(\ZC (\rho_{A}))\leq \log \ddd .
\end{align}
The first line notes that the Holevo quantity is monotonic under channels, $\chi(\{ p_j, \rho_j  \} )\geq \chi(\{ p_j, \EC( \rho_j ) \} )$ \cite{NieChu00}. The second line notes that the entropy is invariant to unitaries giving $S(\rho'_{AB,j}) = S(\rho_{AB})$. The third lines notes that a decohering channel such as $\ZC$ never decreases the entropy, giving $S(\rho'_{AB})\leq S((\ZC \ot \IC)(\rho'_{AB}))=S((\ZC \ot \IC)(\rho_{AB}))$, where the last equality is from \eqref{eqn9}. For the fourth line we let $C$ be a quantum system that purifies $\rho_{AB}$, which gives $S((\ZC \ot \IC)(\rho_{AB})) = S(\ZC(\rho_{A}))+\sum_k q_k S(\rho_{B,k})= S(\ZC(\rho_{A}))+\sum_k q_k S(\rho_{C,k})=S((\ZC \ot \IC)(\rho_{AC}))$, where $\rho_{C,k}$ is defined analogously to $\rho_{B,k}$ and hence they have the same non-zero spectrum. For the last line we invoked subadditivity for the state $(\ZC \ot \IC)(\rho_{AC})$.\end{proof}

\section{Proof of Lemma~\ref{thmEqProbs}} \label{app1}

\begin{proof}
Let us first mention a few basic properties that we will use. Any MES $\ket{\phi}\in \HC_{AB}$ can be written as $\ket{\phi} = (\id \ot U_B)\ket{\phi_0}$ where $\ket{\phi_0} = \sum_k \ket{k}\ket{k}/\sqrt{\ddd}$ is the standard Bell state and $U_B$ is some local unitary. For any operator $Q$, we have 
\begin{equation}
\label{eqnMESprop1}
(Q\ot \id) \ket{\phi_0} =(\id \ot Q^T )\ket{\phi_0}
\end{equation}
where $^T$ denotes the transpose in the standard basis.

Another important property is that the depolarising channel commutes with the unitary encodings: for any unitary $U$, 
\begin{equation}
\label{eqnDepChan11}
\EC^{\bt}_d(U \rho U\ad) = U \EC^{\bt}_d( \rho)U\ad.
\end{equation}

Now let us write the encoding index as $j = (m,n)$ where $m$ and $n$ run from 0 to $d-1$, and denote the ensemble received by Bob as $\{ p_{(m,n)}, \rho''_{AB,(m,n)} \}$, with 
$$\rho''_{AB,(m,n)} =( \EC^{\bt}_d \ot \IC)[(\sgt_X^m\sg_Z^n\ot \id)\rho^{\al}_{AB}(\sgt_X^m\sg_Z^n\ot \id)\ad].$$
The ensemble average state is
\begin{align}
\label{eqn13456833}
&\rho''_{AB} = \sum_{m,n} p_{(m,n)} \rho''_{AB,(m,n)}\notag\\
& = \sum_{m,n} p_{(m,n)} (\EC^{\al\bt}_d \ot \IC)  [(\sgt_X^m\sg_Z^n\ot \id)\dya{\phi}(\sgt_X^m\sg_Z^n\ot \id)\ad] \notag\\
& \eqprop{S}{=} \sum_{(m,n)} p_{(m,n)} (\EC^{\al\bt}_d \ot \IC)  [(\sgt_X^m\ot \sg_Z^n)\dya{\phi_0}(\sgt_X^m\ot \sg_Z^n)\ad] \notag\\
& = \sum_{m,n} p_{(m,n)} \sgt_X^m\ot \sg_Z^n (\EC^{\al\bt}_d \ot \IC) (\dya{\phi_0})(\sgt_X^m\ot \sg_Z^n)\ad\notag\\
& \eqprop{S}{=}  \sum_{m,n} p_{(m+p,n+q)} \sgt_X^{m}\ot \sg_Z^{n} (\EC^{\al\bt}_d \ot \IC) (\dya{\phi_0}) (\sgt_X^{m}\ot \sg_Z^{n})\ad
\end{align}
where the notation $\eqprop{S}{=}$ means equal up to the von Neumann entropy, and we invoked properties \eqref{eqnMESprop1} and \eqref{eqnDepChan11}. The last line in \eqref{eqn13456833} acted on the state with a unitary $(\sgt_X^p\ot \sg_Z^q)\ad$ (which preserves entropy), hence the indices of the probability distribution got shifted forward by $(p,q)$. Now denote the last line in \eqref{eqn13456833} as $\sg_{AB,(p,q)}$, then
\begin{align}
\label{eqn13456834}
&\sum_{p,q}(1/d^2) \sg_{AB,(p,q)} \notag\\
& =  \sum_{m,n} (1/d^2) \sgt_X^{m}\ot \sg_Z^{n} (\EC^{\al\bt}_d \ot \IC) (\dya{\phi_0})(\sgt_X^{m}\ot \sg_Z^{n})\ad \notag\\
& \eqprop{S}{=}  \sum_{m,n} (1/d^2) \rho''_{AB,(m,n)} \, .
\end{align}
Now the concavity of entropy gives
\begin{align}
\label{eqn13456835}
S \Big[ \sum_{m,n} (1/d^2) \rho''_{AB,(m,n)} \Big] &= S \Big[\sum_{p,q}(1/d^2) \sg_{AB,(p,q)} \Big] \notag\\
&\geq  \sum_{p,q} (1/d^2) S [\sg_{AB,(p,q)} ] \notag\\
&= S \Big[\sum_{m,n} p_{(m,n)} \rho''_{AB,(m,n)} \Big] .
\end{align}
Therefore, we have 
\begin{align}
\label{eqn13456836}
&\chi ( \{p_{(m,n)}, \rho''_{AB,(m,n)} \} ) \notag\\
&\leq S \Big[ \sum_{m,n} (1/d^2) \rho''_{AB,(m,n)} \Big] - \sum_{m,n} p_{(m,n)}S[\rho''_{AB,(m,n)}]\notag\\
&= S \Big[ \sum_{m,n} (1/d^2) \rho''_{AB,(m,n)} \Big] - S(\rhoh_{AB})\notag\\
&=\chi ( \{1/d^2, \rho''_{AB,(m,n)} \} ),
\end{align}
which proves the desired result. Here we used the fact that the depolarising channel commutes with the encodings and hence $S [\rho''_{AB,(m,n)}] = S(\rhoh_{AB}), \forall (m,n)$, where $\rhoh_{AB} = (\EC^{\bt}_d \ot \IC)(\rho^{\al}_{AB})$.
\end{proof}

\section{Proof of Theorem~\ref{thmWerDep1}} \label{app2}

\begin{proof}
For a generic Pauli operator $\sg_W = \sum_k \om^k \dya{w_k}$, the action of a (uniformly) random unitary channel where the unitaries are the powers of $\sg_W$ can be rewritten as:
\begin{align}
\label{eqn1345683}
&(1/d)\sum_m \sg_W^m (\cdot) (\sg_W^m)\ad \notag\\
&= (1/d) \sum_{m,k,k'} \om^{(k-k')m}\dya{w_k} (\cdot) \dya{w_{k'}} \notag\\
&=\sum_{k} \dya{w_k} (\cdot) \dya{w_{k}}
\end{align}
since $(1/d)\sum_m \om^{(k-k')m} =\dl_{k,k'}$. Using this, we have (for the case of encoding with equal probabilities)
\begin{align}
&S(\rho''_{AB}) \notag\\
&= S \Big[ \EC_d^{\al\bt} \Big( (1/d^2)\sum_{m,n} \sgt_X^m\sg_Z^n \dya{\phi} (\sgt_X^m\sg_Z^n)\ad \Big) \Big]\notag\\ 
&= S\Big[ \EC_d^{\al\bt} \Big(\sum_{k,l} \ket{x_l}\ip{x_l}{k}\ip{k}{\phi}\ip{\phi}{k}\ip{k}{x_l}\bra{x_l} \Big)\Big] \notag\\ 
&= S\Big[ \EC_d^{\al\bt} \Big(\sum_{k,l} (c_{kl}/d) \dya{x_l} \ot \dya{w_k} \Big) \Big] \notag\\ 
&= S\Big[ \sum_{k,l} (\al\bt c_{kl}/d+(1-\al\bt)/d^2) \dya{x_l} \ot \dya{w_k}   \Big]\notag\\ 
\label{eqn1345688}&= H\Big[ \{ \frac{\al\bt c_{kl}}{d} +\frac{1-\al\bt }{d^2} \}_{kl} \Big].
\end{align}
For simplicity we dropped the identity operator on $B$ for operators, such as the $\EC_d^{\al\bt}$, that only act non-trivially on $A$. Here, the $\{\ket{w_k}\}$ form an orthonormal basis on $\HC_B$ since $\ket{\phi}$ is maximally entangled. Plugging \eqref{eqn1345688} into \eqref{eqnWerDepEqProb}, and using $S(\rho''_{AB,j}) = S(\rhoh_{AB}), \forall j$, proves the desired result.

\end{proof}

\section{Proof of Prop.~\ref{thm56}} \label{app3aaa}

In what follows, one can imagine rewriting the index as $j = (m,n)$ for Pauli-product encodings. The proof of Prop.~\ref{thm56} is aided by the following lemma.

\begin{lemma}
In the case of equal probabilities $p_j  =1/\ddd^2$ for each of the Pauli product unitaries $U_j=\sgt_X^m \sg_Z^n$, 
\begin{equation}
\label{eqn13}
\rho'_{AB} \leq c \left(\id \ot \rho_B\right).
\end{equation}\end{lemma}
\begin{proof}
In Eq.~\eqref{eqn1345683} we noted that a random unitary channel of Pauli operators can be rewritten as a dephasing channel. So we have
\begin{align}
\label{eqn142398}
 \rho'_{AB} &= \frac{1}{\ddd^2}\sum_{m,n} \left(\sgt_X^m \sg_Z^n \ot \id \right) \rho_{AB}\left( (\sg_Z^n)\ad (\sgt_X^m)\ad \ot \id \right)\notag\\
 &= \sum_{k,l}\left( \ket{x_l}\ip{x_l}{k}\bra{k} \ot \id \right) \rho_{AB} \left(\ket{k}\ip{k}{x_l}\bra{x_l} \ot \id \right) \notag\\
 &= \sum_{k,l} |\ip{x_l}{k}|^2 \dya{x_l}\ot \Tr_A\left[ (\dya{k}\ot \id) \rho_{AB}\right] \notag\\
 &\leq c \sum_{k,l}  \dya{x_l}\ot \Tr_A\left[ (\dya{k}\ot \id) \rho_{AB}\right] \notag\\
 &= c \left(\id \ot \rho_B\right).
\end{align}\end{proof}

Now we give the proof of Prop.~\ref{thm56}.

\begin{proof}
Since the entropy is invariant to unitaries we have $S(\rho'_{AB,j}) = S(\rho_{AB})$ and hence
\begin{align}
\label{eqn11}
\CC(\{U_j\},\rho_{AB},\IC) & =\max_{ \{p_j\} } S(\rho'_{AB})-S(\rho_{AB}).
\end{align}
We obtain a lower bound using equal probabilities $p_j = 1/\ddd^2$ for each of the unitaries $U_j=\sgt_X^m \sg_Z^n$. Since the $\log$ is an operator monotone, then from \eqref{eqn13} we have
\begin{align}
S(\rho'_{AB})&\geq - \Tr_{AB}\left[ \rho'_{AB} \log [c \left(\id \ot \rho_B\right)] \right] \notag\\
&= -\log c  - \Tr_B (\rho_B \log \rho_B).
\end{align}
Combining this with \eqref{eqn11} proves \eqref{eqn15}.

Now Eq.~\eqref{eqn15} allows us to give an alternative proof of \eqref{eqn10}. Following Bowen \cite{BowenPRA.63.022302}, we upper bound the capacity using the subadditivity of von Neumann entropy:
\begin{align}
\label{eqn12223412}
\CC(\{U_j\},\rho_{AB},\IC)&\leq \max_{ \{p_j\} }S(\rho'_A) +S(\rho_B) -S(\rho_{AB})\notag\\
&\leq \log \ddd - S(A|B)_{\rho}.
\end{align}
By setting $c = 1/d$ in \eqref{eqn15}, we see that the upper bound in \eqref{eqn12223412} is achievable, proving \eqref{eqn10}.
\end{proof}

\section{Proof of Prop.~\ref{prop5}} \label{app3}

\begin{proof}
Any MES can be written as $\ket{\phi} = (\id \ot U_B)\ket{\phi_0}$; see the discussion surrounding Eq.~\eqref{eqnMESprop1}. Thus Alice encoding with unitary $U_j$ results in the state $\rho'_{AB,j} = (U_j \ot \id)\rho^{\al}_{AB}(U_j\ad \ot \id) = (\id \ot \tilde{U}) \sg^{\al}_{AB}(\id \ot \tilde{U}\ad)$, where $\sg^{\al}_{AB} = \al \dya{\phi_0} +(1-\al) \id / \ddd^2$ and $\tilde{U} = U_BU_j^T$ is a unitary. Since the action of this unitary $\tilde{U}$ on $B$ commutes with the action of the channel $\EC$ on $A$, the entropy of $\rho''_{AB,j} =( \EC \ot \IC)(\rho'_{AB,j})$ is given by $S(\rho''_{AB,j}) = S[(\EC\ot \IC)(\sg^{\al}_{AB} )] = S(\rhoh_{AB})$, where we used the invariance of the entropy under unitaries first for $\tilde{U}$ and then for $U_B$. Thus, we can write the capacity as
\begin{align}
\label{eqn21}
\CC(\{U_j\}, \rho^{\al}_{AB},\EC)&= \max_{ \{p_j\} } [S(\rho''_{AB})] - S(\rhoh_{AB}).
\end{align}

We obtain an upper bound using subadditivity:
\begin{align}
\label{eqn22}
\CC(\{U_j\}, \rho^{\al}_{AB},\EC)&\leq \max_{ \{p_j\} } [S(\rho''_{A})] +S(\rho_B) - S(\rhoh_{AB})\notag\\
&=  S(\EC(\id/\ddd)) +S(\rho_B) - S(\rhoh_{AB})\notag\\
&=  \log \ddd + S(\rhoh_A) - S(\rhoh_{AB})
\end{align}
since $\rho''_{A}  = \EC(\id/\ddd) = \rhoh_A$ and $\rho_B = \id/\ddd$. 

We obtain a lower bound by setting all probabilities equal to $1/\ddd^2$ for the unitaries $\{\sgt_X^m \sg_Z^n\}$. In this case, Eq.~\eqref{eqn13} implies that $\rho''_{AB}\leq c \cdot\EC(\id)\ot \rho_B$, since $\EC$ preserves positivity. Thus, 
\begin{align}
\label{eqn23bdbd}
S(\rho''_{AB})&\geq -\Tr_{AB}\left[ \rho''_{AB} \log [c \cdot \EC(\id) \ot \rho_B] \right]\notag\\
 &=- \log (\ddd\cdot c)+S(\EC(\id/\ddd)) +S(\rho_B) \notag\\
 &= \log (1/c)+S(\rhoh_A).
\end{align}
Plugging \eqref{eqn23bdbd} into \eqref{eqn21} proves \eqref{eqn26}. Now \eqref{eqn26} implies that the upper bound in \eqref{eqn22} is achievable by setting $c = 1/d$, which proves \eqref{eqn25}.
\end{proof}

\section{Proof of Thm.~\ref{thm58}} \label{app4}

\begin{proof}
To obtain a lower bound, we set all probabilities equal $p_j = 1/\ddd^2$ and use $\rho''_{AB}\leq c \cdot\EC(\id)\ot \rho_B$ from \eqref{eqn13}, which implies that 
\begin{align}
\label{eqn31}S(\rho''_{AB})&\geq -\Tr[\rho''_{AB}\log(c\cdot \EC(\id)\ot \rho_B)] \notag\\
&= S(\rho_B)-\log c -\Tr[\rho''_{A}\log \EC(\id)].
\end{align}
Noting that $S(\rho''_{A}) + D(\rho''_{A} || \EC(\id) ) = -\Tr[\rho''_{A}\log \EC(\id)] $ and inserting \eqref{eqn31} into \eqref{eqn7} gives the desired result, \eqref{eqn29}.
\end{proof}

\end{document}